\documentclass[12pt,draftclsnofoot,onecolumn]{IEEEtran}

\usepackage{setspace}
\usepackage{amsmath,amsfonts}
\usepackage{amsthm}
\usepackage{graphicx}
\usepackage{setspace}

\usepackage{subfigure}

\usepackage{algorithm}
\usepackage{algpseudocode}

\usepackage{cite}
\usepackage{stmaryrd}

\newtheorem{proposition}{Proposition}

\begin{document}

\title{Maximizing Mobile Coverage via Optimal
Deployment of Base Station and Relays
}

\author{Xu Li,  Dongning Guo, John Grosspietsch, Huarui Yin, Guo Wei
\thanks{This work was supported in part by Motorola Center for Seamless Communications at Northwestern University.}
\thanks{X. Li is with Department of Electrical Engineering and Information Science,
University of Science and Technology of China,
Hefei, Anhui 230027, China and he is a visiting Ph.D. student at Northwestern University from September 2013 to August 2014 supported by China Scholarship Council. Email: lixu@mail.ustc.edu.cn.  D. Guo is with Department of Electrical Engineering and Computer Science, Northwestern University, Evanston, IL 60201, USA. Email: dGuo@northwestern.edu. J. Grosspietsch was with Enterprise Mobility Solutions Research at
Motorola Solutions, Inc., Schaumburg, IL 60196, USA. Email: john.k.grosspietsch@ieee.org. H. Yin and G. Wei are with
Department of Electrical Engineering and Information Science,
University of Science and Technology of China,
Hefei, Anhui 230027, China.
Email: yhr@ustc.edu.cn, wei@ustc.edu.cn.}}

\maketitle

\begin{abstract}
Deploying relays and/or mobile base stations is a major means of extending the
coverage of a wireless network.
This paper presents models, analytical results, and algorithms to
answer two related questions:
The first is where to deploy relays  in order to extend the reach
from a base station to the maximum;
the second is where to deploy a mobile base station and how many relays
are needed to reach any point in a given area.
Simple time-division and frequency-division scheduling schemes as well
as an end-to-end data rate requirement are assumed.
An important use case of the results is in the Public Safety Broadband
Network, in which deploying relays and mobile base stations is often crucial to provide coverage to an incident scene.
\end{abstract}

\section{Introduction}
Even with today's seemingly ubiquitous wireless access, many areas and corners are not fully covered by existing networks.  There is often no cellular connection in the basement level of large buildings and in remote unpopulated areas.  The existing cellular infrastructure may also be knocked out of service for periods of times in areas hit by disasters \cite{cite_1,cite_1_2,cite_1_3,cite_1_4,cite_1_5}.  Relays can be used to extend the wireless coverage of, e.g., a cellular network \cite{00new,00}.  To extend the coverage of the Public Safety Broadband Network\footnote{The Public Safety Broadband Network, conceived to be a single, connected, universal network for all public safety purposes, is currently being planned and tested in the U.S.} with manageable cost, it has been proposed that mobile base stations are sent to incident scenes along with public safety personnel \cite{ref62}. An important question to consider is how to optimally deploy relays and mobile base stations.

The base station, relays and the destination form a multi-hop device-to-device (D2D) network to extend the wireless coverage\cite{datarate1,datarate3,datarate4,multi-hop}. Most literatures have studied the conditions where the communication range of all links are independent. When each guard can monitor unlimited range but has no vision through the wall, reference \cite{X3} studied  the deployment of the guards  to cover an art gallery with $n$ walls, i.e.,  an $n$-vertex polygon  which is nonconvex in general. When each sensor can monitor an arbitrary limited range, the optimal deployment patterns for full coverage and $k$-connectivity were studied in \cite{new1,new2, new3}.  The numerical deployment algorithms with a minimum number of sensors to provide full coverage were  discussed in \cite{algcover,new4,new5, new6, new7, new8}.  When the transmit range of each device is limited by  the allocated energy, average traffic data and lifetime requirement, the optimal positions of the base station  and the relays for maximizing the system lifetime were considered in \cite{lifetime_1,lifetime_2,lifetime_3}.

In a practical multi-hop D2D network, the communication range of all links are interrelated and mutually determined by the resources allocation scheme and the  quality of service (QoS) requirement. With fixed distance between the base station and the destination, the optimal positions of the relays minimize the end-to-end outage probability in \cite{fixdistance_outage_1,fixdistance_outage_2,fixdistance_outage_3}, or maximize the end-to-end data rate in \cite{fixdistance_datarate_1}. 
Few papers  have studied the optimal positions of the relays for maximizing the distance between the base station and the destination subject to a QoS requirement. Reference \cite{flexdistance_outage_1} studied the case with a single relay and 
an outage probability requirement. This paper studies the general case with multiple relays and an end-to-end data rate requirement. The resources are shared using either time-division or frequency-division scheduling scheme, and both the deployment of relays and a mobile base station are discussed.

The remainder of this paper is organized as follows. Section \ref{section3}  introduces the network models. Section \ref{sectiondatarate}
 introduces the channel model.

Section \ref{section4} studies the deployment of  relays. Specifically, the relays are deployed between the  base station and the destination. We analytically determine the optimal positions of the relays, so that the reach to the destination is maximized subject to the end-to-end data rate  requirement.  When the number of  relays is small, the maximum reach increases with the number of  relays. But beyond a certain number, deploying more relays does not provide further improvement (in fact, it decreases the maximum reach if the same QoS needs to be maintained).

Section \ref{section5} studies the deployment of a mobile base station. Specifically, the mobile base station is deployed to cover an arbitrary polygon,  which may or may not be convex. Both the general case where the base station can be located anywhere and the situation where the base station is constrained to be outside or on the boundary of the polygon are considered.  We propose efficient algorithms to compute the optimal position of the  mobile base station, so that the minimum signal-to-noise ratio (SNR) of any point over the entire region is maximized. If the polygon  can not be covered by the mobile base station alone, relays are deployed to extend the  coverage. The goal here is limited to reaching any point in the region, rather than covering the entire region at the same time.

Section \ref{section6} shows numerical results.  Concluding remarks are given in Section \ref{section7}.

\section{Network Models}\label{section3}
\subsection{Network Model with Relays}\label{section3A}
The network model with relays is shown in Fig. \ref{SYSTEMMODEL1}, where $K-1$ relays are deployed between the  base station and the destination.  At each relay, signals from both directions are fully decoded and re-transmitted, so that the base station,  relays and the destination form a two-way multi-hop decode-and-forward (DF) D2D network.

\begin{figure}[]
\centering
\includegraphics[width=0.7\textwidth]{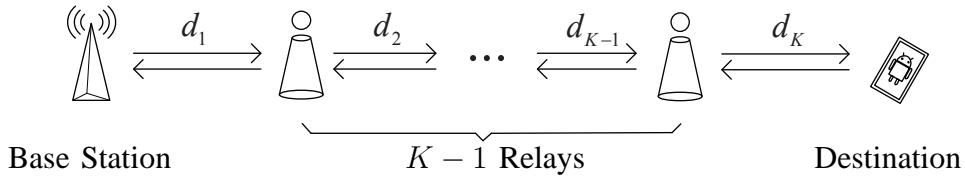}
\put(-350,-10){Base Station}
\put(-45,-10){Destination}
\put(-200,-10){$K-1$ Relays}
\caption{Network model with  relays.}
\label{SYSTEMMODEL1}
\end{figure}

Without loss of generality, we let  relays be located along the line segment that connects the base station and the destination. The $K+1$ devices are connected by $K$ segments and the length of the $k\text{-th}$ segment is $d_k$ in meters for $k\in\{1,\dots,K\}$. The   transmit power of  the $k\text{-th}$ forward and backward links are  $p_k$ and $q_k$. The total bandwidth is $W$. The end-to-end data rate requirement of the forward and backward links are $B$ and $C$, respectively.

We study where the relays should be deployed, so that the reach to the destination is maximized subject to the end-to-end data rate requirement.

\subsection{Network Model with a Mobile Base Station}\label{section3B}
The  network model with a mobile base station is shown in Fig. \ref{SYSTEMMODEL2}, where the mobile  base station is deployed to cover an arbitrary polygon, which may or may not be convex. We first consider the case where the mobile base station can be deployed anywhere. We then consider the case where the base station is constrained to be outside or on the boundary of the polygon, which could be the situation at an incident scene.

In the case that only the mobile base station is deployed, we study where it should be deployed, so that the minimum SNR of any point in the entire region is maximized. Due to the end-to-end data rate requirement, part  of the polygon  may be beyond the coverage of the mobile  base station. In that case, relays are deployed to extend the coverage, where  the optimal positions of the mobile  base station and the relays need to be determined.

\begin{figure}[]
\centering
\includegraphics[width=0.45\textwidth]{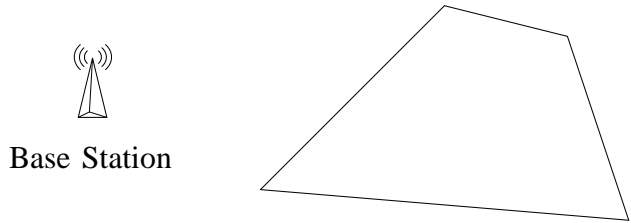}
\put(-235,20){Base Station}
\caption{Network model with a mobile base station.}
\label{SYSTEMMODEL2}
\end{figure}

\section{Channel Model}\label{sectiondatarate}
In the DF  network, an outage event is declared whenever the SNR of any link falls below a prescribed threshold. In other words, the end-to-end data rate is dominated by the weakest link.

\subsection{Fading Model}
We use Rayleigh fading model for links between nearby devices \cite{ref9}. Let the channel power gain of the $k\text{-th}$ segment be:
\begin{align}
h=Ad_k^{-\alpha}\phi
\end{align}
where $A$ is a constant value that considers shadowing and antenna gain,  $\alpha$ is the path loss exponent, $\phi$ denotes the power gain of the Rayleigh fading channel, which follows the exponential distribution with unit mean, $d_k$ is the distance introduced in Section \ref{section3A}.
\subsection{End-to-End Data Rate}\label{datarateB}
We consider uniform time-division and frequency-division scheduling schemes for allocating the resources, which are simple and guarantee robustness
much needed in the Public Safety Broadband Network. Understanding of these simple schemes is also the first step toward more sophisticated solutions involving spectrum sharing.

Under uniform time division, each link uses a $\frac{1}{2K}$ fraction of the time and the total bandwidth. The SNR of the $k\text{-th}$ forward and backward links defined as $\gamma_k$ and $\tau_k$ are:
\begin{subequations}
\begin{alignat}{3}
\gamma_k&=\frac{p_kA\phi}{d_k^{\alpha}W\sigma^2}\\
\tau_k&=\frac{q_kA\phi}{d_k^{\alpha}W\sigma^2}
\end{alignat}
\end{subequations}
where $\sigma^2$ is the power spectral density of the white Gaussian noise, $p_k$ and $q_k$ are the   transmit power  and  W is the bandwidth  introduced in Section \ref{section3A}.

The end-to-end outage probability of the forward and backward links are \cite{fixdistance_outage_3}:
\begin{subequations}
\begin{alignat}{3}
\text{P}[\min(\gamma_1,\dots, \gamma_K )\leq \gamma]&=1-\prod_{k=1}^{K}\exp\left(-\frac{\gamma d_k^\alpha W \sigma^2}{p_kA}\right)\\
\text{P}[\min(\tau_1,\dots, \tau_K )\leq \tau]&=1-\prod_{k=1}^{K}\exp\left(-\frac{\tau d_k^\alpha W \sigma^2}{q_kA}\right)
\end{alignat}
\end{subequations}
where $\gamma$ and $\tau$ are the prescribed SNR threshold of the forward and backward links.

The end-to-end data rate of the forward and backward links under time division defined as $B^{\text{t}}$ and $C^{\text{t}}$ are:
\begin{subequations}\label{000}
\begin{alignat}{3}
B^{\text{t}}&=\frac{W}{2K}\log(1+\gamma)\text{P}[\min(\gamma_1,\dots, \gamma_K )\geq \gamma]=\frac{W}{2K}\log(1+\gamma)\prod_{k=1}^{K}\exp\left(-\frac{\gamma d_k^\alpha W \sigma^2}{p_kA}\right)\\
C^{\text{t}}&=\frac{W}{2K}\log(1+\tau)\text{P}[\min(\tau_1,\dots, \tau_K )\geq \tau]=\frac{W}{2K}\log(1+\tau)\prod_{k=1}^{K}\exp\left(-\frac{\tau d_k^\alpha W \sigma^2}{q_kA}\right).
\end{alignat}
\end{subequations}

Under uniform frequency division, each link uses a $\frac{1}{2K}$ fraction of the bandwidth and the total time. Similarly, the end-to-end data rate of the forward and backward links under frequency division defined as $B^{\text{f}} $ and $C^{\text{f}}$ are:
\begin{subequations}\label{001}
\begin{alignat}{3}
B^{\text{f}}&=\frac{W}{2K}\log(1+\gamma)\prod_{k=1}^{K}\exp\left(-\frac{\gamma d_k^\alpha W \sigma^2}{2Kp_kA}\right)\\
C^{\text{f}}&=\frac{W}{2K}\log(1+\tau)\prod_{k=1}^{K}\exp\left(-\frac{\tau d_k^\alpha W \sigma^2}{2Kq_kA}\right)
\end{alignat}
\end{subequations}
whose  exponential parts  are $\frac{1}{2K}$  of those in  (\ref{000}).

Under power constraint,  the end-to-end data rate achieved according to (\ref{001}) under frequency division are higher than those achieved according to (\ref{000}) under time division. Under energy constraint, the end-to-end  data rate achieved according to (\ref{001}) under frequency division are identical to those achieved according to (\ref{000}) under time division with   $p_k$ and $q_k$ replaced by $2Kp_k$ and $2Kq_k$, respectively.

In the amplify-and-forward (AF) network, each relay amplifies and forwards the signals from both directions. Under the same resources allocation schemes, each link sees no interference. But the noise in previous links is amplified, so that the SNR in AF network decreases with the number of links. Then the end-to-end SNR is smaller than that of the weakest link in AF network, which is also smaller than that of the weakest link in DF network. Hence DF outperforms AF as far as maximizing the reach is concerned.

\section{The Deployment of Relays}\label{section4}
In the case of the deployment of  relays, $K-1$ relays are deployed between the base station and the destination. The optimal distance tuple $(d_1,\dots,d_K)$ maximizes the reach to the destination subject to the end-to-end data rate requirement. Once the optimal distance tuple is computed, the optimal positions of the relays  are straightforward.

\subsection{The General Cases}
Under time division, the end-to-end data rate expressed as  (\ref{000}) should be no less than the data rate requirement $B,C$ introduced in Section \ref{section3A}. The optimization problem  to find the optimal distance tuple $(d_1,\dots,d_K)$ is:
\begin{subequations}\label{cmpTDMA}
\begin{alignat}{3}
\mathop{\text{maximize}}_{\substack{d_1,\dots,d_K}} \ \ &\sum_{k=1}^{K}d_k\\
\text{subject to\ \  }&\frac{W}{2K}\log(1+\gamma)\prod_{k=1}^{K}\exp\left(-\frac{\gamma d_k^\alpha W \sigma^2}{p_kA}\right) \geq B\\&\frac{W}{2K}\log(1+\tau)\prod_{k=1}^{K}\exp\left(-\frac{\tau d_k^\alpha W \sigma^2}{q_kA}\right)\geq C
\end{alignat}
\end{subequations}
which can be rewritten as:
\begin{subequations}\label{0}
\begin{alignat}{3}
\mathop{\text{maximize}}_{\substack{d_1,\dots,d_K}} \ \ &\sum_{k=1}^{K}d_k\\
\text{subject to\ \  }&\sum_{k=1}^{K}\frac{d_k^\alpha}{p_k}\leq b\label{0b}\\
&\sum_{k=1}^{K}\frac{d_k^\alpha}{q_k}\leq c\label{0c}
\end{alignat}
\end{subequations}
where
\begin{subequations}\label{0000}
\begin{alignat}{3}
b=&\frac{A}{\gamma W \sigma^2}\log\left(\frac{W\log(1+\gamma)}{2KB}\right)\label{0000a}\\
c=&\frac{A}{\tau W \sigma^2}\log\left(\frac{W\log(1+\tau)}{2KC}\right).\label{0000b}
\end{alignat}
\end{subequations}

The optimization problem (\ref{0}) is convex and the associated \emph{Lagrangian} is \cite{convex}:
\begin{align}\label{01}
L(d_1,\dots,d_K,\lambda,\nu)=-\sum_{k=1}^{K}d_k+\lambda\left(\sum_{k=1}^{K}\frac{d_k^\alpha}{p_k}-b\right)+\nu\left(\sum_{k=1}^{K}\frac{d_k^\alpha}{q_k}-c\right)
\end{align}
where $\lambda$ and $\nu$ are the \emph{Lagrangian} multipliers.

Let $d_1^*,\dots,d_K^*$, $\lambda^*$ and $\nu^*$ be the optimal solutions. The \emph{Karush-Kuhn-Tucker} (KKT) conditions are:
\begin{subequations}\label{03}
\begin{alignat}{3}
\lambda^* &\geq 0 \label{03a}\\
\nu^* &\geq 0 \label{03b}\\
\sum_{k=1}^{K}\frac{(d_k^{*})^\alpha}{p_k}-b &\leq 0 \label{03c}\\
\sum_{k=1}^{K}\frac{(d_k^{*})^\alpha}{q_k}-c &\leq 0 \label{03d}\\
\lambda^*\left(\sum_{k=1}^{K}\frac{(d_k^{*})^\alpha}{p_k}-b\right)&=0 \label{03e}\\
\nu^*\left(\sum_{k=1}^{K}\frac{(d_k^{*})^\alpha}{q_k}-c\right) &=0 \label{03f}\\
\frac{\partial L(d_1^*,\dots,d_K^*,\lambda^*,\nu^*)}{\partial d_k^*}&=0, k = 1,\dots, K. \label{03g}
\end{alignat}
\end{subequations}
The solutions of (\ref{03g}) are:
\begin{align}\label{04}
d_k^*=\left(\frac{p_kq_k}{\alpha\lambda^* q_k+\alpha \nu^* p_k}\right)^{\frac{1}{\alpha-1}}, k = 1,\dots, K.
\end{align}

We find all feasible distance tuples $(d_1, \dots, d_K)$ satisfying the KKT conditions (\ref{03}) and the one with the maximum sum distance is the optimal distance tuple. The  solutions of the KKT conditions (\ref{03}) have three possible cases.

The first case is that $\lambda^*=0$ and $\nu^*>0$. Using (\ref{03f}) and (\ref{04}), $\nu^*$ satisfies
\begin{align}\label{05}
\sum_{k=1}^{K}\frac{q_k^{\frac{1}{\alpha-1}}}{(\alpha \nu^*)^{\frac{\alpha}{\alpha-1}}}=c.
\end{align}
The left side of (\ref{05})  decreases with $\nu^*$, so that it is easy to obtain the unique solution of (\ref{05}). Then the distance tuple is computed by (\ref{04}). If the distance tuple satisfies (\ref{03c}), it is a feasible solution of the KKT conditions.

The second case is that $\lambda^*>0$ and $\nu^*=0$. Using (\ref{03e}) and (\ref{04}), $\lambda^*$ satisfies
\begin{align}\label{06}
\sum_{k=1}^{K}\frac{p_k^{\frac{1}{\alpha-1}}}{(\alpha \lambda^*)^{\frac{\alpha}{\alpha-1}}}=b.
\end{align}
Similarly, it is easy to obtain the unique solution of (\ref{06}). Then the distance tuple is computed by (\ref{04}). If the distance tuple satisfies (\ref{03d}), it is a feasible solution of the KKT conditions.

The third case is that $\lambda^*>0$ and $\nu^*>0$.  Using (\ref{03e}), (\ref{03f}) and (\ref{04}), $\lambda^*$ and $\nu^*$ satisfy
\begin{subequations}\label{07}
\begin{alignat}{3}
&\sum_{k=1}^{K}\frac{p_k^{\frac{1}{\alpha-1}}q_k^{\frac{\alpha}{\alpha-1}}}{(\alpha\lambda^* q_k +\alpha\nu^* p_k)^{\frac{\alpha}{\alpha-1}}}=b\label{07a}\\
&\sum_{k=1}^{K}\frac{p_k^{\frac{\alpha}{\alpha-1}}q_k^{\frac{1}{\alpha-1}}}{(\alpha\lambda^* q_k +\alpha\nu^* p_k)^{\frac{\alpha}{\alpha-1}}}=c.\label{07b}
\end{alignat}
\end{subequations}
Functions ({\ref{07a}}) and ({\ref{07b}}) both decrease with $\lambda^*$ and $\nu^*$, so that it is not difficult to obtain the solution. Then the feasible distance tuple satisfying the KKT conditions is computed by (\ref{04}).

\renewcommand{\algorithmicrequire}{\textbf{Input:}}
\renewcommand{\algorithmicensure}{\textbf{Output:}}

\begin{algorithm}[]
\caption{Computing the optimal positions of the relays }
\label{Algorithm0}
\begin{algorithmic}[1]
   \State \textbf{Input: }$A$, $B$, $C$, $W$, $\alpha$, $\sigma^2$, $p_k$, $q_k$, for $k\in\{1,\dots,K\}$.
\State \textbf{Output: }$(d_1,\dots,d_K)$.
   \State Compute the  feasible distance tuple in the case of $\lambda^*=0$ and $\nu^*>0$ using (\ref{03c}), (\ref{04}) and (\ref{05}).
   \State Compute the  feasible distance tuple in the case of $\lambda^*>0$ and $\nu^*=0$ using (\ref{03d}), (\ref{04}) and (\ref{06}).
   \State Compute the  feasible distance tuple in the case of $\lambda^*>0$ and $\nu^*>0$ using  (\ref{04}) and (\ref{07}).
\State Compare all feasible distance tuples and  \textbf{return} the one with the maximum sum distance.
\end{algorithmic}
\end{algorithm}

The numerical method for computing the optimal distance tuple for maximizing the reach to the destination is summarized as Algorithm \ref{Algorithm0}.
If either $\lambda^*=0$ or $\nu^*=0$, the solution of  (\ref{05}) or (\ref{06}) is unique, so that there is no more than one feasible distance tuple. If $\lambda^*>0$ and $\nu^*>0$, the solution of (\ref{07}) with non-identical parameters is typically unique or null, so is the number of the feasible distance tuple. Then the complexity of Algorithm \ref{Algorithm0} is the comparison of no more than $3$ feasible distance tuples.

Under frequency division, the end-to-end data rate expressed as  (\ref{001}) should be no less than the data rate requirement. The optimization problem  to find the optimal distance tuple $(d_1,\dots,d_K)$ is:
\begin{subequations}\label{0f}
\begin{alignat}{3}
\mathop{\text{maximize}}_{\substack{d_1,\dots,d_K}} \ \ &\sum_{k=1}^{K}d_k\label{0fa}\\
\text{subject to\ \  }&\sum_{k=1}^{K}\frac{d_k^\alpha}{p_k}\leq 2Kb\label{0fb}\\
&\sum_{k=1}^{K}\frac{d_k^\alpha}{q_k}\leq 2Kc\label{0fc}
\end{alignat}
\end{subequations}
which is convex and similar to (\ref{0}). Then following the same way as  Algorithm \ref{Algorithm0}, the optimal distance tuple for maximizing the reach to the destination is computed.

\subsection{The Cases with Identical Parameters}
In the following, we study the cases with identical parameters, which shed more light on the intuition about how to optimally deploy  relays for maximizing the reach to the destination.  The identical parameters are $p_k=q_k=p$, $\gamma=\tau$ and $B=C$ for $k\in\{1,\dots,K\}$.

\begin{proposition}
\label{proposition:0}
Under time division, the optimal distance of each link is $\left(\frac{bp}{K}\right)^{\frac{1}{\alpha}}$ and the maximum reach to  the destination is ${({bp})^{\frac{1}{\alpha}}K^{\frac{\alpha-1}{\alpha}}}$. Under frequency division, the optimal distance of each link is $(2bp)^{\frac{1}{\alpha}}$ and the maximum reach to  the destination is $({2bp})^{\frac{1}{\alpha}}K$.
\end{proposition}
\begin{proof}Under time division, the constraints (\ref{0b}) and (\ref{0c}) are identical. The KKT conditions become:
\begin{subequations}\label{08new}
\begin{alignat}{3}
\lambda^* &\geq 0 \label{08newa}\\
\sum_{k=1}^{K}\frac{(d_k^{*})^\alpha}{p}-b &\leq 0 \label{08newb}\\
\lambda^*\left(\sum_{k=1}^{K}\frac{(d_k^{*})^\alpha}{p}-b\right)&=0 \label{08newc}\\
d_k^*&=\left(\frac{p}{\alpha\lambda^*}\right)^{\frac{1}{\alpha-1}}, k = 1,\dots, K\label{08newd}
\end{alignat}
\end{subequations}
whose feasible solution is unique with $\lambda^* > 0$. The optimal distance of each link and the maximum reach to  the destination are: 
\begin{subequations}\label{08}
\begin{alignat}{3}
&d_1^*=,\dots,=d_K^*=\left(\frac{bp}{K}\right)^{\frac{1}{\alpha}}\label{08a}\\
&\sum_{k=1}^{K}d_k^*={({bp})^{\frac{1}{\alpha}}K^{\frac{\alpha-1}{\alpha}}}.\label{08b}
\end{alignat}
\end{subequations}

Under frequency division, the optimization problem (\ref{0f}) uses $2Kb$ instead of $b$ in the optimization problem (\ref{0}). Then using   $2Kb$ instead of $b$ in (\ref{08a}) and (\ref{08b}), the optimal distance of each link and the maximum reach to  the destination are: 
\begin{subequations}\label{09}
\begin{alignat}{3}
&d_1^*=,\dots,=d_K^*=(2bp)^{\frac{1}{\alpha}}\label{09a}\\
&\sum_{k=1}^{K}d_k^*=({2bp})^{\frac{1}{\alpha}}K.\label{09b}
\end{alignat}
\end{subequations}
\end{proof}

\begin{proposition}
\label{proposition:01} Let  $\beta$ be $\frac{W\log(1+\gamma)}{2B}$.
 Under time division, the optimal distance of each link decreases with $K$, and the maximum reach to  the destination increases with $K$ for $K\leq \lfloor\beta \exp\left(\frac{1}{1-\alpha}\right)\rfloor$, but decreases for  $K\geq \lceil\beta \exp\left(\frac{1}{1-\alpha}\right)\rceil$. Under frequency division, the optimal distance of each link decreases with $K$, and the maximum reach to  the destination increases with $K$ for $K\leq \lfloor\beta \exp\left(-\frac{1}{\alpha}\right)\rfloor$, but decreases for $K\geq \lceil\beta \exp\left(-\frac{1}{\alpha}\right)\rceil$.
\begin{proof}
Under time division, substituting (\ref{0000}) into (\ref{08}), the optimal distance of each link is:
\begin{align}\label{010}
\left(\frac{bp}{K}\right)^{\frac{1}{\alpha}}=\left(\frac{Ap}{\gamma W \sigma^2K}\log\left(\frac{\beta}{K}\right)\right)^{\frac{1}{\alpha}}
\end{align}
which decreases with $K$. The maximum reach to the destination is:
\begin{align}\label{011}
({bp})^{\frac{1}{\alpha}}K^{\frac{\alpha-1}{\alpha}}=\left(\frac{Ap}{\gamma W \sigma^2}\log\left(\frac{\beta}{K}\right)\right)^{\frac{1}{\alpha}}K^{\frac{\alpha-1}{\alpha}}
\end{align}
which increases with $K$ for $K\leq \lfloor\beta\exp\left(\frac{1}{1-\alpha}\right)\rfloor$, but decreases for $K\geq \lceil\beta \exp\left(\frac{1}{1-\alpha}\right)\rceil$.

Under frequency division, substituting (\ref{0000}) into (\ref{09}), the optimal distance of each link is:
\begin{align}\label{012}
({2bp})^{\frac{1}{\alpha}}=\left(\frac{2Ap}{\gamma W \sigma^2}\log\left(\frac{\beta}{K}\right)\right)^{\frac{1}{\alpha}}
\end{align}
which decreases with $K$. The maximum reach to the destination is:
\begin{align}\label{013}
({2bp})^{\frac{1}{\alpha}}K=\left(\frac{2Ap}{\gamma W \sigma^2}\log\left(\frac{\beta}{K}\right)\right)^{\frac{1}{\alpha}}K
\end{align}
which increases with $K$ for $K\leq \lfloor\beta\exp\left(-\frac{1}{\alpha}\right)\rfloor$, but decreases for $K\geq \lceil\beta \exp\left(-\frac{1}{\alpha}\right)\rceil$.
\end{proof}
\end{proposition}

Therefore, when the number of  relays is small, more relays increase the maximum reach to  the destination. Beyond a certain number, deploying more relays not only does not provide further improvement, but decreases the maximum reach.

\subsection{Discussions}\label{Discuss}
In the following, we discuss the upper bound of the number of relays and the traffic types for general cases.
\subsubsection{The Number of Relays}
The fraction of the resources allocated to each link decreases with the number of  relays, which may decrease the end-to-end data rate. In order to satisfy the end-to-end data rate requirement, the number of  relays should be bounded.

\begin{proposition}
\label{proposition:0}
Let $\zeta$ be $\frac{W\log(1+\tau)}{2C}$ and the number of  relays is upper bounded by:
\begin{align}\label{00000}
K-1\leq\min\left(\lfloor\beta-1\rfloor,\lfloor\zeta-1\rfloor\right).
\end{align}
\end{proposition}
\begin{proof} For the optimization problem (\ref{0}) and (\ref{0f}), the parameters defined as $b$ and $c$ should be positive,
\begin{subequations}\label{014}
\begin{alignat}{3}
b&=\frac{A}{\gamma W \sigma^2}\log\left(\frac{\beta}{K}\right)\geq 0\label{014a}\\
c&=\frac{A}{\tau W \sigma^2}\log\left(\frac{\zeta}{K}\right)\geq 0.\label{014b}
\end{alignat}
\end{subequations}
Hence (\ref{00000}) must hold.
\end{proof}

Therefore, each additional relay adds to the maximum reach to  the destination in general. But there must exist a certain number, beyond which deploying more relays is not necessary and does not provide further improvement in the maximum reach.

\subsubsection{Traffic Types}
For continuous traffic, each link starts transmission when it is scheduled. Regardless of under time division or frequency division, the resources are fully utilized. For bursty traffic, each link starts transmission as soon as the bursty package arrives. Under time division, the time resources could be fully utilized. Under frequency division, only one link is transmitting with a fraction of the spectrum while others are waiting for the bursty package, so that the spectrum resources are inefficiently utilized.

Therefore, according to the resources utilization, time division and frequency division  are identical for continuous traffic, but time division outperforms frequency division for bursty traffic.

\section{The Deployment of a Mobile Base Station}\label{section5}
In the case of the deployment of a mobile base station, the mobile base station is deployed to cover an arbitrary polygon, which may or may not be convex. The optimal position of the mobile base station maximizes the minimum SNR of any point over the entire region. Both the general case where the base station can be located anywhere and the situation where the base station is constrained to be outside or on the boundary of the polygon are considered.

\subsection{Without Place Restrictions}
For isotropic channel and omni-directional antennas, the coverage of the mobile base station is a disk. The optimal position of the mobile  base station for maximizing the minimum SNR of any point over the entire polygon is the center of the minimum disk covering  the polygon.
\begin{proposition}
\label{proposition:1}
The optimal position of the mobile base station for covering the polygon is identical to the optimal position of the mobile base station for covering all vertices of the polygon.
\end{proposition}
\begin{proof}The minimum disk covering the polygon is defined as $\mathcal{D}$, which covers all vertices of the polygon. The  minimum disk covering all vertices of the polygon is defined as  $\mathcal{V}$, which should be proved to cover the entire polygon. Since the vertices are covered by $\mathcal{V}$, the edges of the polygon connecting any two adjacent vertices are covered, and hence also the line segments connecting any two points on the sides. Then the entire polygon is covered, so that $\mathcal{D}$ and  $\mathcal{V}$ are the same disk. Therefore, the optimal position of the mobile  base station for covering the polygon which maximizes the minimum SNR over any point in the entire polygon is identical to the optimal position of the mobile base station for covering all vertices of the polygon which maximizes the minimum SNR of any vertex of the polygon.
\end{proof}

Following  \textit{Proposition \ref{proposition:1}}, we simplify the problem to find the center of the minimum disk covering all vertices  of the polygon.  Let the coordinate of the $M$ vertices of the polygon be $(x_1,y_1),\dots,(x_M,y_M)$. The optimization  problem to find the optimal position of the mobile base station is:
\begin{subequations}\label{c0}
\begin{alignat}{3}
\mathop{\text{minimize}}_{x,y}\ \ \ &r\label{c0a}\\
\text{subject to}\ \ &(x-x_i)^2+(y-y_i)^2\leq r^2,i=1\dots M\label{c0b}
\end{alignat}
\end{subequations}
which is a convex problem and $(x,y)$ is the optimal position of the mobile base station. The optimal position of the mobile base station also maximizes the $N$-of-$N$ system lifetime for homogeneous destinations\footnote{With large number of homogeneous destinations uniformly distributed in the polygon, each destination communicates with the mobile base station. The $N$-of-$N$ system lifetime is the time until any destinations run out of energy.} in \cite{lifetime_1}. Instead of using KKT conditions, we study the structure of the optimal position of the mobile base station to develop a more efficient numerical method.

In order to satisfy the constraint (\ref{c0b}), $(x,y)$ locates in a disk  whose center is the vertex $(x_i,y_i)$ and radius is $r$ for $i\in\{1,\dots,M\}$. If $r$ is too small, the intersection of the $M$ disks centered at $(x_1,y_1),\dots,(x_M,y_M)$ is empty. Let $r$ increase and at the first time there is one point in the intersection, the point is the optimal position of the mobile base station. There  are only two possible cases as far as the geometry of the problem is concerned,  which are illustrated in Fig. \ref{covernew}. The coverage of the mobile base station is a disk, which could be determined by a  diameter in Fig. \ref{fig:subfig0:a} or three points on the boundary in Fig. \ref{fig:subfig0:b}. The optimal position of the mobile base station is marked by a star. We define $\mathcal{C}_i$ as the circle whose center  is the vertex $(x_i,y_i)$ and radius is $r$.

\begin{figure}[]
  \centering
  \subfigure[]{
    \label{fig:subfig0:a} 
    \includegraphics[width=0.3\textwidth]{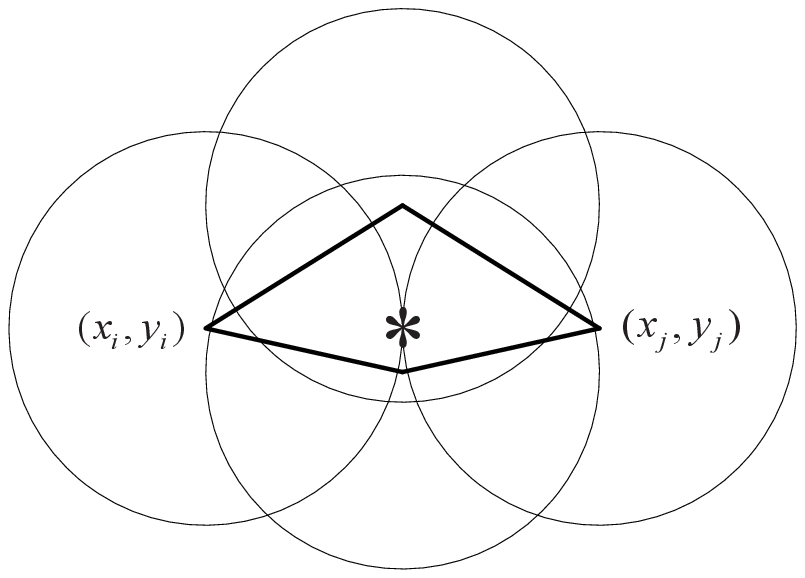}}
  \hspace{0.1in}
  \subfigure[]{
    \label{fig:subfig0:b} 
    \includegraphics[width=0.3\textwidth]{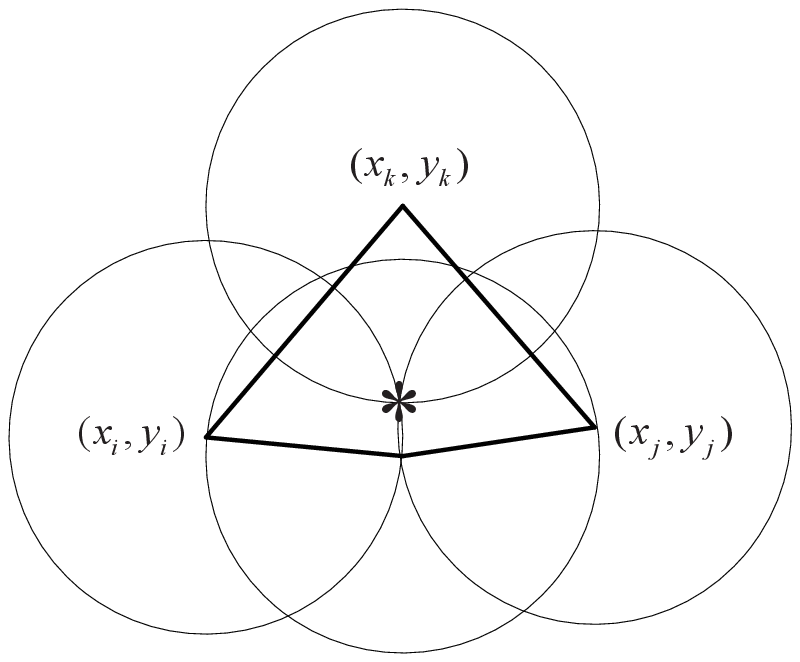}}
  \caption{The optimal position of the mobile base station without place restrictions.}
  \label{covernew} 
\end{figure}

In Fig. \ref{fig:subfig0:a}, the optimal position of the mobile  base station $(x,y)$ is the sole intersection point of two circles $\mathcal{C}_i$ and $\mathcal{C}_j$. The line segment connecting the vertices $(x_i, y_i)$ and $(x_j, y_j)$ is a diameter of the coverage of the mobile base station. For two  indexes $i, j\in\{1,\dots, M\}$, the coordinate of the mobile base station is:
\begin{subequations}\label{c1}
\begin{alignat}{3}
x&=\frac{x_i+x_j}{2}\label{c1a}\\
y&=\frac{y_i+y_j}{2}.\label{c2b}
\end{alignat}
\end{subequations}

In Fig. \ref{fig:subfig0:b}, the optimal position of the  mobile base station is the sole intersection point of three circles $\mathcal{C}_i$, $\mathcal{C}_j$ and $\mathcal{C}_k$, which is also the center of the circumcircle of the triangle whose vertices are $(x_i, y_i)$, $(x_j, y_j)$ and $(x_k, y_k)$. For three different indexes $i, j, k\in\{1,\dots, M\}$, the coordinate of the mobile base station is:
\begin{subequations}\label{c3}
\begin{alignat}{3}
x&=\frac{(x_i^2+y_i^2)(y_j-y_k)+(x_j^2+y_j^2)(y_k-y_i)+(x_k^2+y_k^2)(y_i-y_j)}{2x_i(y_j-y_k)+2x_j(y_k-y_i)+2x_k(y_i-y_j)}\label{c3a}\\
y&=\frac{(x_i^2+y_i^2)(x_k-x_j)+(x_j^2+y_j^2)(x_i-x_k)+(x_k^2+y_k^2)(x_j-x_i)}{2x_i(y_j-y_k)+2x_j(y_k-y_i)+2x_k(y_i-y_j)}.\label{c3b}
\end{alignat}
\end{subequations}

We find all candidate positions described by (\ref{c1}) and (\ref{c3}), and compute the maximum distance to all vertices for each candidate position as:
\begin{align}\label{c4}
r=\max_{i\in\{1,\dots,M\}}\sqrt{(x-x_i)^2-(y-y_i)^2}.
\end{align}
The candidate position with the minimum maximum distance to all vertices is the optimal position of the mobile base station.

\begin{algorithm}[]
\caption{Computing the optimal position of the mobile base station without place restrictions}
\label{Algorithmc0}
\begin{algorithmic}[1]
   \State \textbf{Input: }$(x_i,y_i)$ for $i\in\{1,\dots,M\}$.
\State \textbf{Output: }$(x,y)$.
   \State Find the candidate positions using (\ref{c1}), (\ref{c3}).
   \State Compute the maximum distance to all vertices for each candidate position using (\ref{c4}).
\State \textbf{Return} $(x,y)$ which is the coordinate of the candidate position with the minimum maximum distance to all vertices.
\end{algorithmic}
\end{algorithm}

The numerical method computes the optimal  position of the mobile base station without place constrictions as Algorithm \ref{Algorithmc0}. The complexity is the comparison of no more than $\frac{1}{6}M^3-\frac{1}{6}M$ candidate positions satisfying (\ref{c1}) or (\ref{c3}) where $M\geq 3$.

\subsection{With Place Restrictions}
When the mobile base station shall be deployed outside or on the boundary of the polygon, the optimization problem is changed to:
\begin{subequations}\label{24}
\begin{alignat}{3}
\mathop{\text{minimize}}_{x,y}\ \ \ &r\label{24a}\\
\text{subject to}\ \ &(x-x_i)^2+(y-y_i)^2\leq r^2,i=1\dots M\label{24b}\\
&(x,y)\notin \mathcal{I}\label{24c}
\end{alignat}
\end{subequations}
which is a nonconvex problem in general and $\mathcal{I}$ stands for the interior of the polygon.

\begin{figure}[]
  \centering
  \subfigure[]{
    \label{fig:subfig:c} 
    \includegraphics[width=0.3\textwidth]{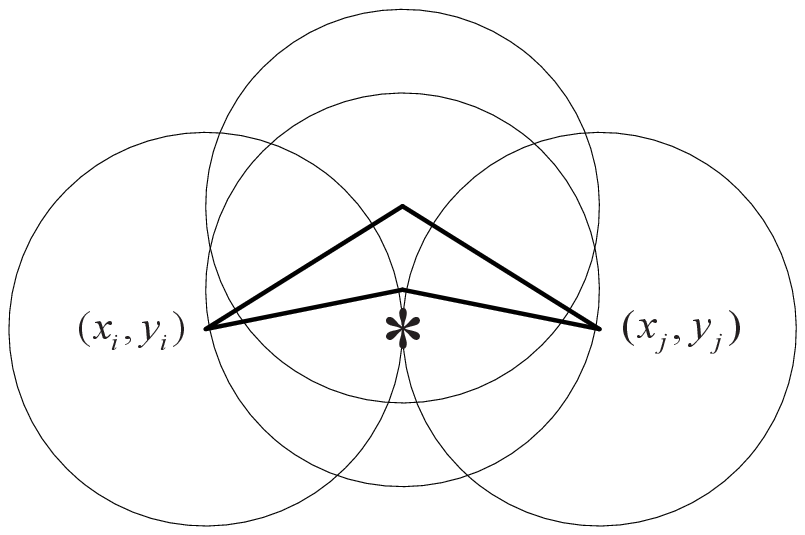}}
  \hspace{0.1in}
  \subfigure[]{
    \label{fig:subfig:d} 
    \includegraphics[width=0.3\textwidth]{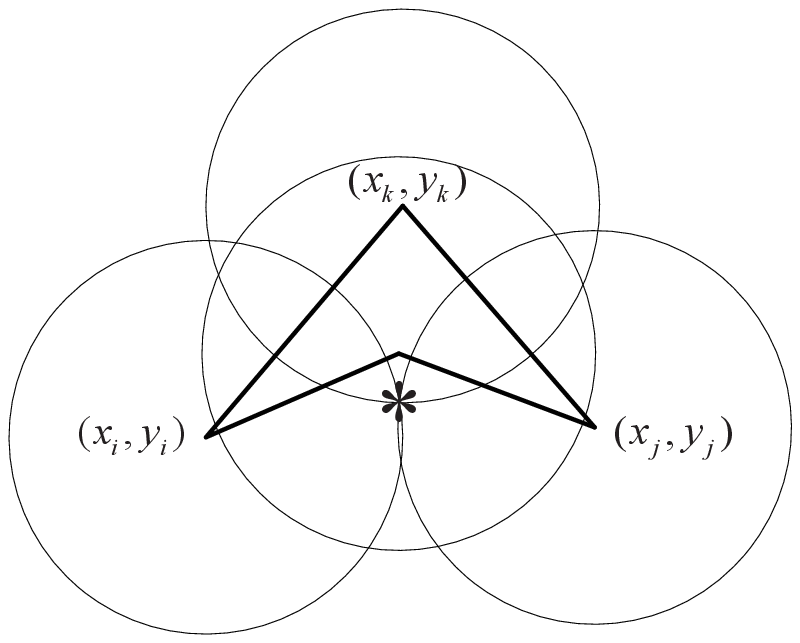}}
  \hspace{0.1in} \\
  \subfigure[]{
    \label{fig:subfig:a} 
    \includegraphics[width=0.3\textwidth]{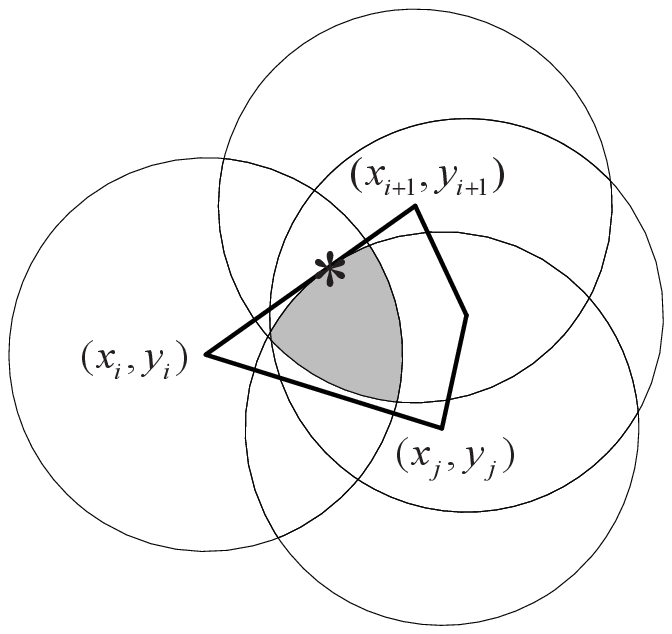}}
  \hspace{0.1in}
  \subfigure[]{
    \label{fig:subfig:b} 
    \includegraphics[width=0.3\textwidth]{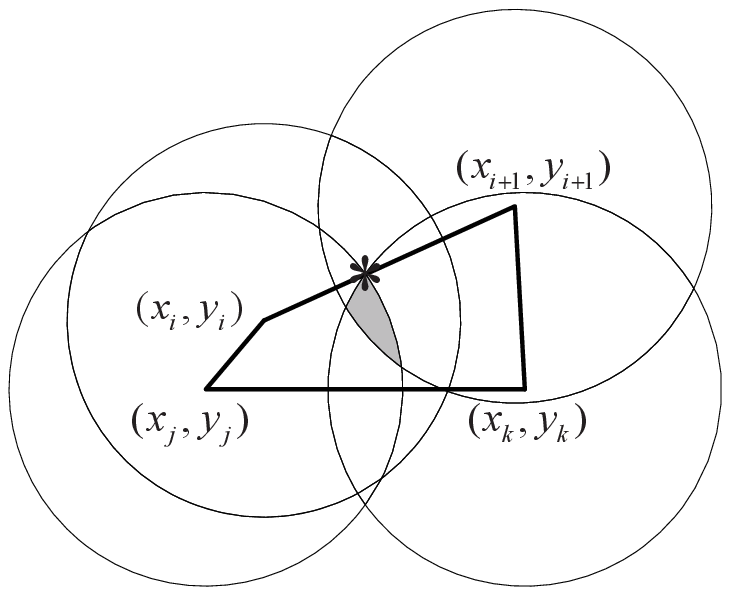}}
  \caption{The optimal position of the mobile base station outside or on the boundary of the polygon.}
  \label{geometry} 
\end{figure}

Similarly, let $r$ increases and at the first time the intersection area of the $M$ disks  centered at $(x_1,y_1),\dots,(x_M,y_M)$ has one point outside or on the boundary of the polygon, the point is the optimal position of the mobile base station. There are only four possible cases which are illustrated in Fig. \ref{geometry}.
The intersection area is the shadowed  part and the optimal position of the mobile base station is marked by a star.  In Fig. \ref{fig:subfig:c} and Fig. \ref{fig:subfig:d}, the  optimal position of the mobile  base station is outside the polygon. In Fig. \ref{fig:subfig:a} and Fig. \ref{fig:subfig:b}, the optimal position of the mobile base station is on the boundary  of the polygon. We define $\mathcal{E}_{i}$ as the edge  of the polygon  connecting two adjacent vertices $(x_i,y_i)$ and $(x_{i+1},y_{i+1})$ for $i\in\{1,\dots,M\}$ where  $(x_{M+1},y_{M+1})=(x_{1},y_{1})$.

In Fig. \ref{fig:subfig:c}, the optimal position of the mobile  base station $(x,y)$ is the sole intersection point of two circles $\mathcal{C}_i$ and $\mathcal{C}_j$, and outside the polygon.  For two non-adjacent indexes $i, j\in\{1,\dots, M\}$, the coordinate of the mobile base station is described by  (\ref{c1}).

In Fig. \ref{fig:subfig:d}, the optimal position of the mobile base station $(x,y)$ is the sole intersection point of three circles $\mathcal{C}_i$, $\mathcal{C}_j$ and $\mathcal{C}_k$, and outside the polygon.  For three different indexes $i, j, k\in\{1,\dots, M\}$, the coordinate of the mobile base station is  described by (\ref{c3}).

In Fig. \ref{fig:subfig:a}, the optimal position of the mobile base station $(x,y)$ is the intersection point of one circle $\mathcal{C}_j$ and one edge $\mathcal{E}_{i}$. The line segment connecting $(x_j, y_j)$ and $(x,y)$ is orthogonal to  $\mathcal{E}_{i}$ and $(x,y)$ is on $\mathcal{E}_{i}$. For three different indexes
$i, i+1, j\in\{1,\dots, M\}$, the coordinate of the mobile base station is:
\begin{subequations}\label{25}
\begin{alignat}{3}
x&=\frac{x_j(x_{i+1}-x_i)^2+(y_{i+1}-y_i)(x_iy_{i+1}+x_{i+1}y_j-x_iy_j-x_{i+1}y_i)}{(x_{i+1}-x_i)^2+(y_{i+1}-y_i)^2}\label{25a}\\
y&=\frac{y_j(y_{i+1}-y_i)^2+(x_{i+1}-x_i)(x_{i+1}y_i+x_jy_{i+1}-x_jy_i-x_iy_{i+1})}{{(x_{i+1}-x_i)^2+(y_{i+1}-y_i)^2}}.\label{25b}
\end{alignat}
\end{subequations}

In Fig. \ref{fig:subfig:b}, the optimal position of the mobile base station $(x,y)$ is the intersection point of two circles $\mathcal{C}_j$ and $\mathcal{C}_k$ and one edge $\mathcal{E}_i$. The distance from $(x,y)$ to $(x_j, y_j)$ and $(x_k, y_k)$ are the same and $(x,y)$ is on $\mathcal{E}_i$. For four indexes   $i, i+1, j, k\in\{1,\dots, M\}$ where  $j\neq k$, the coordinate of the mobile base station is:
\begin{subequations}\label{26}
\begin{alignat}{3}
x&=\frac{(x_{i+1}-x_i)(x_k^2-x_j^2+y_k^2-y_j^2)+2(y_k-y_j)(x_iy_{i+1}-x_{i+1}y_i)}{2(y_k-y_j)(y_{i+1}-y_i)+2(x_k-x_j)(x_{i+1}-x_i)}\label{26a}\\
y&=\frac{(y_{i+1}-y_i)(x_k^2-x_j^2+y_k^2-y_j^2)+2(x_k-x_j)(x_{i+1}y_i-x_iy_{i+1})}{2(y_k-y_j)(y_{i+1}-y_i)+2(x_k-x_j)(x_{i+1}-x_i)}.\label{26b}
\end{alignat}
\end{subequations}

The candidate position with the minimum maximum distance to all vertices is the optimal position of the mobile base station.

\begin{algorithm}[]
\caption{Computing the optimal position of the mobile base station with place restrictions}
\label{Algorithm3}
\begin{algorithmic}[1]
   \State \textbf{Input: }$(x_i,y_i)$ for $i\in\{1,\dots,M\}$.
\State \textbf{Output: }$(x,y)$.
   \State Find the candidate positions using (\ref{c1}), (\ref{c3}) and delete the ones inside the polygon.
   \State Find the candidate positions using (\ref{25}), (\ref{26}) and delete the ones outside the polygon.
   \State Compute the maximum distance to all vertices for each candidate position using (\ref{c4}).
\State \textbf{Return} $(x,y)$ which is the coordinate of the candidate position with the minimum maximum distance to all vertices.
\end{algorithmic}
\end{algorithm}

The numerical method computes the optimal  position of the mobile base station outside or on the boundary of the polygon as Algorithm \ref{Algorithm3}. The complexity is the comparison of no more than $\frac{2}{3}M^3+\frac{1}{2}M^2-\frac{19}{6}M$ candidate positions satisfying (\ref{c1}), (\ref{c3}), (\ref{25}) or (\ref{26}) where $M\geq 3$.

\subsection{Relay-Assisted Coverage}
The optimal position of the mobile base station maximizes the minimum SNR of any point over the entire region. But due to the end-to-end data rate requirement, part  of the polygon  may be beyond the coverage of the mobile base station.  When the destination locates at the area beyond the coverage of the mobile base station, the minimum number of relays are deployed along the line segment connecting the mobile base station and  the destination.

Using Algorithm \ref{Algorithm0}, we first compute the optimal positions of the relays  for maximizing the reach to the destination  as the number of  relays increases. Once the maximum reach is larger than or equal to the distance between the mobile base station and the destination, the number of  relays is the minimum one and the relays are deployed at the optimal positions.

We caution that, with relays, the goal here is limited to reaching any point in the region, rather than covering the entire region at the same time.
\begin{figure}
\centering
\includegraphics[width=0.35\textwidth]{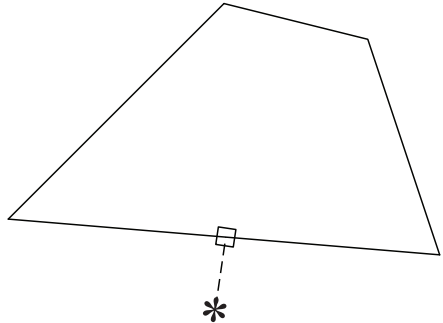}
\caption{The optimal position of the mobile base station to cover the convex polygon. }
\label{fig_convex}
\end{figure}
\subsection{Discussions}
For the nonconvex polygon, the optimal position of the mobile base station could be inside, outside or on the boundary of the polygon.  But the convex polygon  satisfies  \textit{Proposition \ref{proposition:6}}.
\begin{proposition}
\label{proposition:6}
For the convex polygon, the optimal position of the mobile base station can not be outside the polygon.
\end{proposition}
\begin{proof}
If the optimal position of the mobile base station is outside the convex polygon  as  marked by a star in Fig. \ref{fig_convex},  there exists at least one edge  of the polygon,  such that
the star and the polygon  locate at the different sides of the edge.  We draw an orthogonal line from the star to the edge  and mark the intersection point by a square. The maximum distance to all vertices from the square is smaller than that from the star, so that the square is a better position of the  base station which contradicts the assumption.  Therefore,  the optimal position of the mobile base station can not be outside the convex polygon.
\end{proof}

Therefore, for the convex polygon without place restrictions, the optimal position of the mobile base station is inside or on the boundary of the polygon.
For the convex polygon with place restrictions, the optimal position of the mobile base station is  on the boundary of the polygon.

\section{Numerical Results}\label{section6}

The simulation parameters chosen according to LTE standards \cite{ref10} are listed in TABLE \ref{parameters}.

\begin{table}
\centering
\caption{Main Simulation Parameters}
\label{parameters}
\renewcommand{\arraystretch}{1.0}
\begin{tabular}{p{0.23\textwidth}|p{0.26\textwidth}}
\hline
Parameter&Value\\
\hline
Transmit power of the base station&$P$ dBm\\
\hline
Transmit power of the relays&$P-3$ dBm\\
\hline
Transmit power of the destination&$P-6$ dBm\\
\hline
Total bandwidth $W$&$9$ MHz\\
\hline
Data rate requirement $B,C$&$2$ Mbps\\
\hline
SNR threshold $\gamma,\tau$&$20$ dB\\
\hline
Path loss $Ad^{-\alpha}$&$-15.3-37.6\log_{10}(d)$ dB, $d$ in meters\\
\hline
Power spectral density of Gaussian noise $\sigma^2$&$-174$ dBm/Hz\\
\hline
\end{tabular}
\end{table}

\subsection{The Deployment of Relays}
In the case of the deployment of  relays,  the optimal distance tuple is computed which maximizes the reach to the destination subject to the end-to-end data rate requirement.

The optimal distance tuple as the   transmit power increases with one relay under time division and frequency division are shown in Fig. \ref{dvsp}. With fixed   transmit power, $d_1$ is larger than $d_2$ since the   transmit power of the base station is larger than that of the destination. As discussed in Section \ref{datarateB}, when the   transmit power of each device is the constraint, the end-to-end data rate under time division is smaller than that under frequency division. Then $d_1+d_2$ under time division is smaller than that under frequency division. With the optimal distance tuple, the optimal positions of the relays  are straightforward. For example, when $P=20$ dBm, the maximum reach to  the destination is $341$ meters under time division and $493$ meters under frequency division. The optimal position of the relay is $192$ meters  from the  base station under time division and $277$ meters  from the  base station under frequency division.

\begin{figure}[]
\centering
\includegraphics[width=0.6\textwidth]{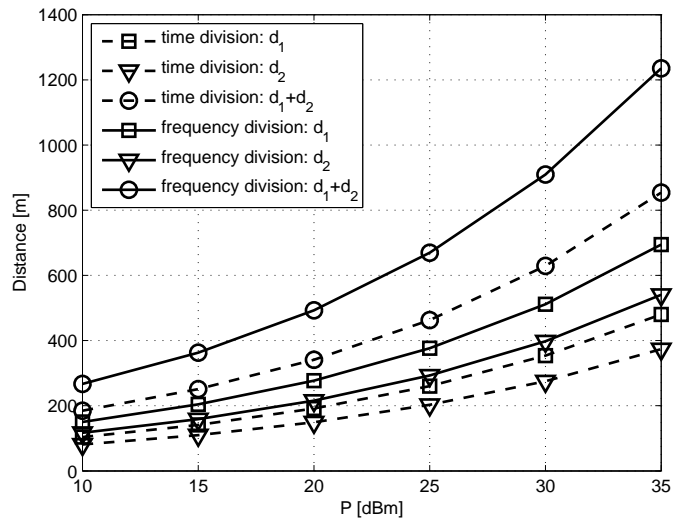}
\caption{The optimal distance tuple as the   transmit power increases with one relay under time division and frequency division.}
\label{dvsp}
\end{figure}

\begin{figure}[]
\centering
\includegraphics[width=0.6\textwidth]{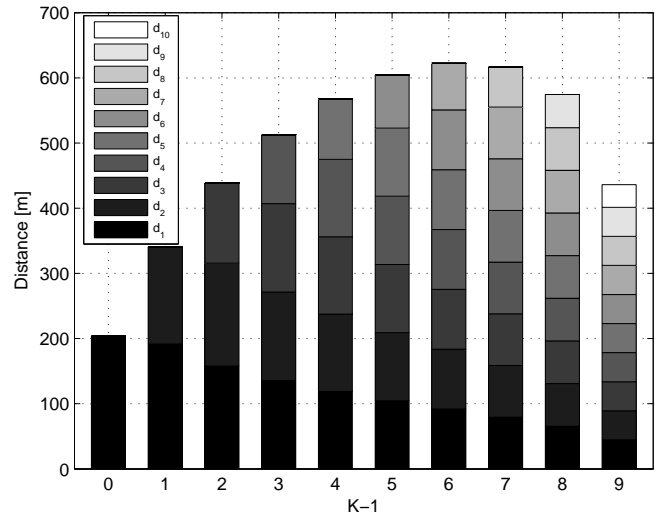}
\caption{The optimal distance tuple as the number of  relays increases under time division when $P=20$ dBm.}
\label{dvsntdma}
\end{figure}

\begin{figure}[]
\centering
\includegraphics[width=0.6\textwidth]{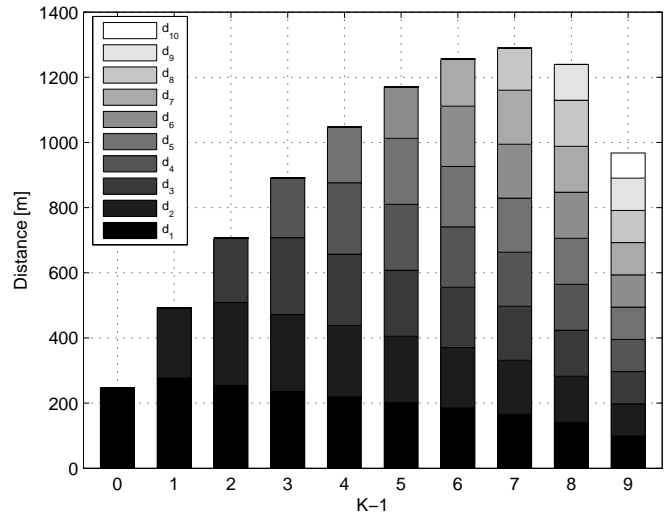}
\caption{The optimal distance tuple  as the number of  relays increases under frequency division when $P=20$ dBm.}
\label{dvsnfdma}
\end{figure}

As \textit{Proposition \ref{proposition:0}} shows, substituting the simulation parameters into (\ref{00000}),  the number of  relays $K-1$ can not exceed $9$.
The optimal distance tuple  as the number of  relays increases  under time division and frequency division when $P=20$ dBm are shown in Fig. \ref{dvsntdma} and Fig. \ref{dvsnfdma}. When the number of  relays  is small, the maximum reach to  the destination increases with the  number of relays. Beyond $6$ under time division and $7$ under frequency division,  deploying more relays not only does not provide further improvement, but decreases the maximum reach. Even with redundant relays, the best choice is to deploy $6$ relays to extend the  maximum reach to  $622$ meters under time division, and $7$ relays to extend the maximum reach to $1289$ meters under frequency division.


\subsection{The Deployment of the Mobile Base Station}
In the case of the deployment of the mobile base station, the polygon is  a quadrilateral whose four vertices are $(0,350)$, $(300,650)$, $(500,600)$, and $(600,300)$ in meters. The optimal position of the mobile base station is computed, which maximizes the minimum SNR of any point over the entire region.
\begin{figure}[]
\centering
\includegraphics[width=0.6\textwidth]{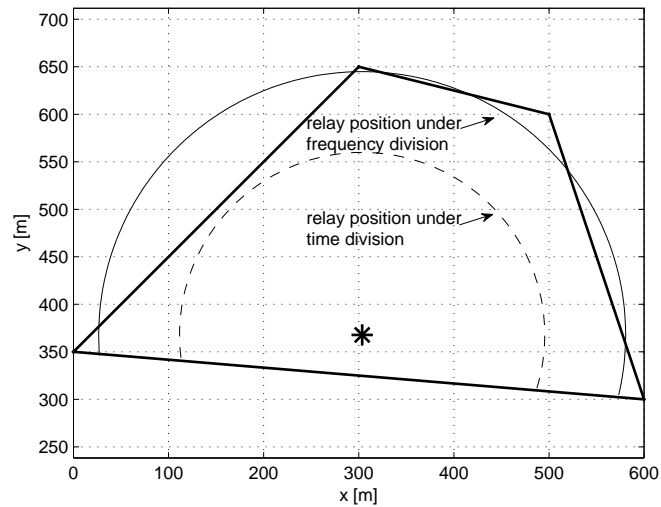}
\caption{The optimal positions of the  mobile base station without place restrictions  and  the relay under time division and frequency division when $P=20$ dBm.}
\label{journalcovernew}
\end{figure}

When the mobile base station could be deployed anywhere, the optimal position of the mobile base station is $(368,304)$ in meters as marked by a star in Fig. \ref{journalcovernew}. The distance from the mobile base station to the vertices  of the polygon  are $304$, $282$, $304$ and $304$ meters.  As shown in Fig. \ref{dvsntdma} and Fig. \ref{dvsnfdma}, with the same end-to-end data rate requirement, the coverage of the mobile base station is $204$ meters under time division and $245$ meters under frequency division when $P=20\text{ dBm}$.
When the destination locates at the area beyond the coverage of the mobile base station, one relay is deployed along the line segment connecting  the mobile base station and the destination to  extend the coverage to $341$ meters under time division, and $493$ meters under frequency division. The relay is on the dash chord which is $192$ meters  from  the mobile  base station under time division, and on the solid chord which is $277$ meters from the mobile base station under frequency division.

\begin{figure}[]
\centering
\includegraphics[width=0.55\textwidth]{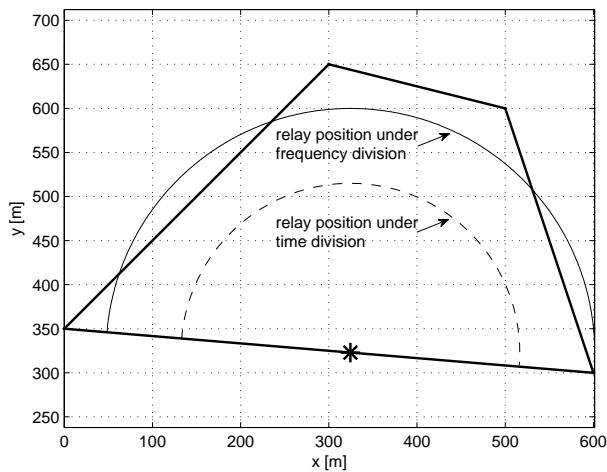}
\caption{The optimal positions of the mobile base station outside or on the boundary of the polygon  and the relay under time division and frequency division when $P=20$ dBm.}
\label{journalcover}
\end{figure}

When the mobile  base station shall be deployed outside or on the boundary of the polygon, the optimal position of the mobile base station is $(325,323)$ in meters as marked by a star in Fig. \ref{journalcover}. The distance from the mobile base station to the vertices  of the polygon  are $326$, $328$, $328$ and $276$ meters.  Similarly,  when the destination locates at the area beyond the coverage of the mobile base station, one relay is deployed on the dash chord which is $192$ meters  from  the mobile base station under time division, and on the solid chord which is $277$ meters from the mobile base station under frequency division.

\section{Conclusion}\label{section7}
In this paper, we study the optimal positions of multiple relays and a mobile base station for maximizing the reach of a wireless network under time division and frequency division. When the number of relays is small, the maximum reach increases with the number of relays. Beyond a certain number, deploying more relays does not provide further improvement due to the end-to-end data rate requirement. The optimal position of the mobile base station maximizes the minimum SNR of any point in a polygon. Due to the end-to-end data rate requirement, relays may be deployed to reach any point in the polygon. Maximizing the reach of the wireless network  using spectrum sharing scheme and providing full coverage to a given area are left for future work, where interference and scheduling are crucial to the wireless coverage.

\end{document}